\documentclass[journal]{IEEEtran}
\IEEEoverridecommandlockouts                              

\usepackage[utf8]{inputenc}
\usepackage{color}
\usepackage{array}
\usepackage{float}
\usepackage{amsmath,amssymb,amsfonts}
\usepackage{verbatim}

\usepackage{graphicx}
\usepackage{booktabs}
\usepackage{balance}
\usepackage{xcolor}
\def\BibTeX{{\rm B\kern-.05em{\sc i\kern-.025em b}\kern-.08em
    T\kern-.1667em\lower.7ex\hbox{E}\kern-.125emX}}
\usepackage{cite}    
\usepackage{hyperref}
\hypersetup{colorlinks=true,linkcolor=red,citecolor=blue,urlcolor=black}  

\usepackage{lipsum}
\usepackage{mathtools}
\usepackage{cuted}

\usepackage{algorithmic}
\usepackage{longtable}

\floatstyle{ruled}
\newfloat{algorithm}{tbp}{loa}
\providecommand{\algorithmname}{Algorithm}
\floatname{algorithm}{\protect\algorithmname}

\floatstyle{ruled}
\newfloat{algorithm}{tbp}{loa}
\providecommand{\algorithmname}{Algorithm}
\floatname{algorithm}{\protect\algorithmname}

\newcommand{\MYfooter}{\smash{
		\hfil\parbox[t][\height][t]{\textwidth}{\centering
			\thepage}\hfil\hbox{}}}

\makeatletter
\let\oldforeign@language\foreign@language
\DeclareRobustCommand{\foreign@language}[1]{%
	\lowercase{\oldforeign@language{#1}}}

\let\oldforeign@language\foreign@language
\DeclareRobustCommand{\foreign@language}[1]{%
	\lowercase{\oldforeign@language{#1}}}

\ifCLASSINFOpdf
\else
\fi

\ifCLASSINFOpdf
\else
\fi

\newtheorem{thm}{Theorem}

\newtheorem{assum}{Assumption}
\ifCLASSINFOpdf
\else
\fi

\def\ps@IEEEtitlepagestyle{%
	\def\@oddhead{\parbox[t][\height][t]{\textwidth}{\centering \scriptsize
			Personal use of this material is permitted. Permission from the author(s) and/or copyright holder(s), must be obtained for all other uses. Please contact us and provide details if you believe this document breaches copyrights.\\
			\noindent\makebox[\linewidth]{}
		}\hfil\hbox{}}%
	\def\@evenhead{\scriptsize\thepage \hfil \leftmark\mbox{}}%
	\def\@oddfoot{\parbox[t][\height][l]{\textwidth}{
			\vspace{-20pt}{\rule{\textwidth}{0.4pt}}\\ \footnotesize{\bf{\footnotesize\textcolor{red}{H. Naser, H. A. Hashim, and M. Ahmadi, "Adaptive Gain Nonlinear Observer for External Wrench Estimation in Human-UAV Physical Interaction," In Proc. of the 2026 IEEE International Conference on Robotics \& Automation (ICRA), Vienna, Austria 2026.}}}\\
			\noindent\makebox[\linewidth]
		}\hfil\hbox{}}%
	\def\@evenfoot{\MYfooter}}

\makeatother
\begin{document}
\bstctlcite{IEEEexample:BSTcontrol}
\pagenumbering{gobble}
\title{\LARGE \bf Adaptive Gain Nonlinear Observer for External Wrench Estimation in Human-UAV Physical Interaction}
\author{Hussein N. Naser, Hashim A. Hashim, and Mojtaba Ahmadi
	\thanks{This work was supported in part by National Sciences and Engineering Research Council of Canada (NSERC), under the grants RGPIN-2022-04937 and the University of Thi-Qar, Iraqi Ministry of Higher Education and Scientific Research under financial support (No. 6608 in 21/06/2022).}
	\thanks{H. N. Naser, H. A. Hashim, and M. Ahmadi are with the Department of Mechanical and Aerospace Engineering, Carleton University, Ottawa, Ontario, K1S-5B6, Canada. Also, H. N. Naser is with the Department of Biomedical Engineering, University of Thi-Qar, Nasiriyah, Thi-Qar, 64001, Iraq. Email: HusseinNaser@cmail.carleton.ca, husseinalshami@utq.edu.iq, hhashim@carleton.ca, and Mojtaba.Ahmadi@carleton.ca}}

\maketitle


\begin{abstract}
This paper presents an Adaptive Gain Nonlinear Observer (AGNO) for estimating the external interaction wrench (forces and torques) in human-UAV physical interaction for assistive payload transportation. The proposed AGNO uses the full nonlinear dynamic model to achieve an accurate and robust wrench estimation without relying on dedicated force-torque sensors. A key feature of this approach is the explicit consideration of the non-constant inertia matrix, which is essential for aerial systems with asymmetric mass distribution or shifting payloads. A comprehensive dynamic model of a cooperative transportation system composed of two quadrotors and a shared payload is derived, and the stability of the observer is rigorously established using Lyapunov-based analysis. Simulation results validate the effectiveness of the proposed observer in enabling intuitive and safe human-UAV interaction. Comparative evaluations demonstrate that the proposed AGNO outperforms an Extended Kalman Filter (EKF) in terms of estimation root mean square errors (RMSE), particularly for torque estimation under nonlinear interaction conditions. This approach reduces system weight and cost by eliminating additional sensing hardware, enhancing practical feasibility.
\end{abstract}

\begin{IEEEkeywords}
Nonlinear observer, Adaptive, Force-torque estimation, sensorless Wrench estimation, Lyapunov stability, Human-UAV Physical interaction, Cooperative payload transportation. 
\end{IEEEkeywords}

\section{INTRODUCTION}
Unmanned Aerial Vehicles (UAVs) are increasingly integrated into human workspaces, necessitating intuitive and safe human-UAV interaction\cite{laghari2023unmanned}. Physical interaction, where humans directly apply forces to UAVs for guidance, offers a promising paradigm for collaborative manipulation and assistive technologies\cite{naser2025aerial}. Accurate measurements or estimations of human-applied forces are fundamental for intuitive and responsive systems\cite{rajappa2017design}.

Traditionally, interaction forces are measured using dedicated force-torque sensors \cite{mariotti2019admittance}. However, these sensors add considerable weight that reduces payload capacity and flight endurance, especially for small to medium-sized UAVs \cite{naser2026QUKF}, increase system complexity and cost due to calibration and protection requirements \cite{ruggiero2018aerial}, and limit interaction flexibility while remaining vulnerable to damage during collisions \cite{tomic2014unified}. These limitations motivate research into estimating interaction forces without direct measurement \cite{tomic2014unified}. Force estimation methods utilize the system's existing state measurements and dynamic model to infer applied forces, offering a more elegant and practical solution \cite{wilmsen2019nonlinear}. Estimation approaches eliminate additional sensors, reducing weight and complexity, and allow interaction at any point on the UAV. They depend on existing navigation sensors (IMU, position tracking) and can be implemented as software modules, making them adaptable and cost-effective \cite{naser2025aerial}.

\subsection{Related Work}
Physical human-robot interaction has been extensively studied for manipulators and humanoid robots, with impedance control being a foundational framework \cite{hogan1984impedance}. Extending these paradigms to aerial vehicles presents unique challenges due to their underactuated nature, inherent instability, and safety concerns \cite{rajappa2017control}. Early work demonstrated physical guidance of the quadrotors \cite{augugliaro2013admittance}. Recent advances focus on intuitive and safe interaction, expanding applications to collaborative payload transportation and aerial manipulation \cite{naser2025aerial,rajappa2017design,lee2018integrated}.

Given the limitations of direct force sensing, researchers have explored various estimation methods. Kalman filter-based approaches, including linear Kalman filters and their nonlinear variants, have been applied for force estimation; however, they face challenges with the nonlinearity of the UAV, linearization errors, increased computational complexity, and the need for careful tuning of noise covariance matrices \cite{skrede2024linear, yin2023error, banks2021physical}. Concurrently, data-driven machine learning approaches have emerged for force estimation, capable of handling complex nonlinearities; however, they demand extensive training data, exhibit limited generalization, and impose high computational demands on small UAVs. In addition, their "black box" nature complicates theoretical guarantees on performance and stability \cite{alharbat2025external, kruvzic2021end, dai2019learning}. In contrast, nonlinear observer approaches directly incorporate system nonlinear dynamics without approximation, demonstrating superior performance during rapid force changes and improved robustness to model uncertainties. These observers handle nonlinearity naturally, adapt better to uncertainties, and provide robust and accurate estimates without heavy dependence on noise covariance tuning. This makes them well suited for estimating external interaction wrenches in dynamic human-UAV interaction scenarios for cooperative payload transportation \cite{wilmsen2019nonlinear, veil2021nonlinear, yuksel2014nonlinear, nikoobin2009lyapunov, hashim2023observer}.

Inspired by the advantages of nonlinear observers, this paper proposes an Adaptive Gain Nonlinear Observer (AGNO) to estimate the external interaction wrench for assistive cooperative aerial payload transportation. This work aims to advance human-UAV physical interaction, enabling more intuitive, safe, and accessible interaction paradigms without specialized force sensing hardware.

\subsection{Contributions}
This paper makes several significant contributions to the field of human-UAV physical interaction:
\begin{enumerate}
    \item A detailed dynamic model derivation for an assistive cooperative aerial payload transportation system, consisting of two quadrotors rigidly attached to a beam-shaped payload.
    \item An adaptive gain nonlinear observer (AGNO) is proposed for estimating external interaction forces and torques.
    \item An explicit analysis of a non-constant inertial matrix is provided, critical for aerial systems with asymmetric mass distribution or shifting payloads.
    \item A rigorous stability analysis using Lyapunov theory establishes guarantees on the observer's convergence and robustness properties.
    \item A benchmark validation of the proposed AGNO is benchmarked against an EKF, demonstrating a lower RMSE and a superior torque estimation under nonlinear interaction conditions.
\end{enumerate}
This paper is structured as follows: Section \ref{sec:methodology} details the system modeling and control allocation and strategy. Section \ref{sec:observer_design} introduces the proposed nonlinear observer design. Section \ref{sec:stability_analysis} presents a comprehensive stability analysis, with a specific focus on non-constant inertial matrices. Section \ref{sec:simulation_results} discusses the simulation results. Finally, Section \ref{sec:conclusion} provides the conclusions of this work.

\section{Methodology\label{sec:methodology}}
\subsection{System Overview and Problem Formulation}
This work presents a cooperative aerial system for intuitive human-guided payload transportation. Two quadrotor UAVs are rigidly connected to a shared payload, enabling a human operator to guide the assembly through 3D space via direct physical interaction. To achieve this, we develop detailed dynamic models for each component and an integrated model for the unified system. Human interaction forces are then estimated using a nonlinear observer, which is detailed in the following section.

\begin{figure}
\centering \includegraphics[clip,width=1\linewidth]{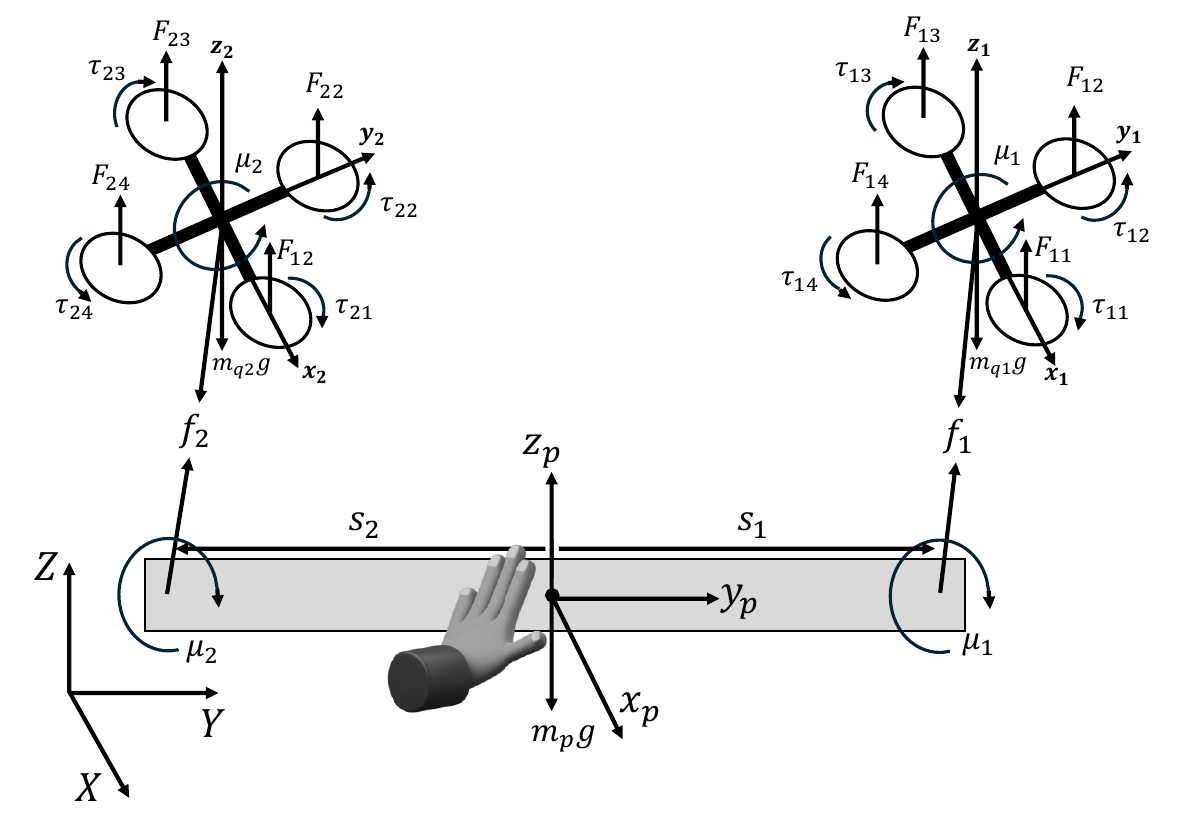}
\caption{Free body diagram of the system components.}
\label{Fig1:system_FBD} 
\end{figure}

\subsection{Reference Frames and Notation}
The reference coordinate frames to describe the dynamics of the system are defined to have a world-inertial frame that is denoted as $W_{I} = \{O, X, Y, Z\}$, while a body-fixed frame $B_{p} = \{o_{p}, x_{p},y_{p},z_{p}\}$ is attached to the payload's center of mass (CoM). Each quadrotor is also associated with its own body-fixed frame $Q_{j} = \{o_{j}, x_{j}, y_{j}, z_{j}\}$, for $j = 1, 2$. All coordinate systems are defined such that their positive $z$-axis points upward. The position of the quadrotor $j$ in the inertial frame is represented by $p_{j} = [x_{j}, y_{j}, z_{j}]^{\top}$. The angular velocities expressed in the body-fixed frame are given by $\omega = [p, q, r]^{\top}$, and the orientation of the system is described using Euler angles $\xi = [\phi, \theta, \psi]^{\top}$. The rotation matrix $R(\xi)$ maps vectors from the body-fixed frame to the inertial frame.

The relationship between angular rates $\omega$ and Euler angle rates $\dot{\xi}$ is given by\cite{naser2025aerial}:
\begin{equation}
\omega = \Theta(\xi)\dot{\xi}, \quad \Theta(\xi) = \begin{bmatrix} c(\theta) & 0 & -s(\theta) \\ 0 & 1 & s(\phi) \\ s(\theta) & 0 & c(\phi) c(\theta) \end{bmatrix},
\label{eq1:omega_psi_dot}
\end{equation}
where $s(\cdot)$ and $c(\cdot)$ represent $\sin(\cdot)$ and $\cos(\cdot)$, respectively.

\subsection{The Payload Dynamics}
From the Free Body Diagram of the system in \ref{Fig1:system_FBD}, the dynamic model of the payload with position $p_{p} =[x_{p},y_{p},z_{p}]^{\top}$, is:
\begin{equation}
\begin{aligned}
m_{p}\ddot{p}_{p}&=(f_{1}+f_{2})-m_{p}g\hat{e}, \\
I_{p}\dot{\omega}&=(\mu_{1}+\mu_{2}) +(s_{1}\times f_{1})+(s_{2}\times f_{2})-\omega\times I_{p}\omega,
\label{eq2:payload_dynamics}
\end{aligned}
\end{equation}
where $m_{p}$ denotes the mass of the payload, $I_{p}$ denotes the inertia matrix, $\ddot{p}_{p}$ is the translational acceleration, $\hat{e}=[0,0,1]^{\top}$ is a unit vector on $z$-axis, $g$ is the gravity, $f_{1}$, $f_{2}$, $\mu_{1}$, and $\mu_{2}$ denote the forces and torques exerted by the quadrotors on the payload, respectively, and $s_{1} $ and $s_{2}$ are position vectors from payload CoM to the quadrotors' CoM.

\subsection{The Quadrotor Dynamics}
The translational and rotational dynamics of each quadrotor according to the Newton-Euler formalism are as follows:
\begin{equation}
\begin{aligned}
m_{q_{j}}\ddot{p}_{q_{j}}&= RT_{q_{j}}\hat{e}-f_{j}-m_{q_{j}}g\hat{e}, \\
I_{q_{j}}\dot{\omega}&=\tau_{q_{j}}-\mu_{j}-\omega\times I_{q_{j}}\omega,
\label{eq3:quad_dynamics}
\end{aligned}
\end{equation}
where $m_{q_{j}}$ is the mass of the $j^{th}$ quadrotor, $I_{q_{j}}$ is the inertia matrix, $\ddot{p}_{q_{j}}$ is the translational acceleration, $\dot{\omega}$ is the angular acceleration, $f_{j}$ and $\mu_{j}$ are internal forces and torques, $
T_{q_{j}}$ and $\tau_{q_{j}}=[\tau_{j_{1}},\tau_{j_{2}},\tau_{j_{3}}]^{\top}$ denote the total thrust and the moment vector of the $j^{th}$ quadrotor, respectively. These control inputs are derived from the rotor thrusts as:
\begin{equation}
\begin{bmatrix}T_{q_{j}}\\ \tau_{q_{j}}\end{bmatrix}=\begin{bmatrix}1&1&1&1\\ 0&r&0&-r\\ -r&0&r&0\\ \varrho&-\varrho&\varrho&-\varrho\end{bmatrix}\begin{bmatrix}F_{j1}\\ F_{j2}\\ F_{j3}\\ F_{j4}\end{bmatrix},
\label{eq4:rotor_inputs}
\end{equation}
where $F_{ji}=k_{t}S_{ji}^{2}$ is the force generated by each rotor and $k_{t}$ is the thrust constant, $S_{ji}$ is the rotor angular speed, $\varrho=k_{m}/k_{t}$ is a ratio of constants, $k_{m}$ is the constant of the drag moment, which is given as $\tau_{ji}=k_{m}S_{ji}^{2}$, and $r$ is the arm length of the $j^{th}$ quadrotor.

\subsection{The Integrated System Dynamics}
Referring to the Free Body Diagram in Figure \ref{Fig1:system_FBD}, the following
assumptions are necessary for the modeling purposes of the integrated aerial payload transportation system. 
\begin{assum}
\label{assum:1} 
    The system comprises three rigid components (two identical quadrotors and one payload) rigidly connected, each with specific specifications and physical characteristics.
\end{assum}
\begin{assum}
\label{assum:2} 
    The payload is a hollow tube with mass $m_{p}$, length $L_p$, and radius $r_{p}$.
\end{assum}
\begin{assum}
\label{assum:3} 
    $x$-axis and $y$-axis in the body-fixed reference frame are the symmetry axes of the system.
\end{assum}

The dynamics model of the integrated system, with $p=[x,y,z]^{\top}$ is the position of its CoM measured in the inertial reference frame, is derived by combining the dynamics in \eqref{eq2:payload_dynamics} and \eqref{eq3:quad_dynamics}, as follows:
\begin{equation}
\begin{aligned}
m\ddot{p}&=RT\hat{e}-mg\hat{e}+f_{ex}, \\
I\dot{\omega}&=\tau-\omega\times I\omega+M_{ex},
\label{eq5:integrated_dynamics}
\end{aligned}
\end{equation}
where $m=m_{q_1}+m_{q_2}+m_{p}$ represents the total mass of the system, $I$ represents the diagonal matrix of the total moment of inertia, $f_{ex} \in \mathbb{R}^{3}$ and $M_{ex} \in \mathbb{R}^{3}$ are the force and torque of the human external interaction, respectively, and $T, \tau$ are the total thrust and moments of the system coming from the two quadrotors.

\subsection{Control Architecture}

\begin{figure}[b!]
\centering
\includegraphics[width=1\columnwidth]{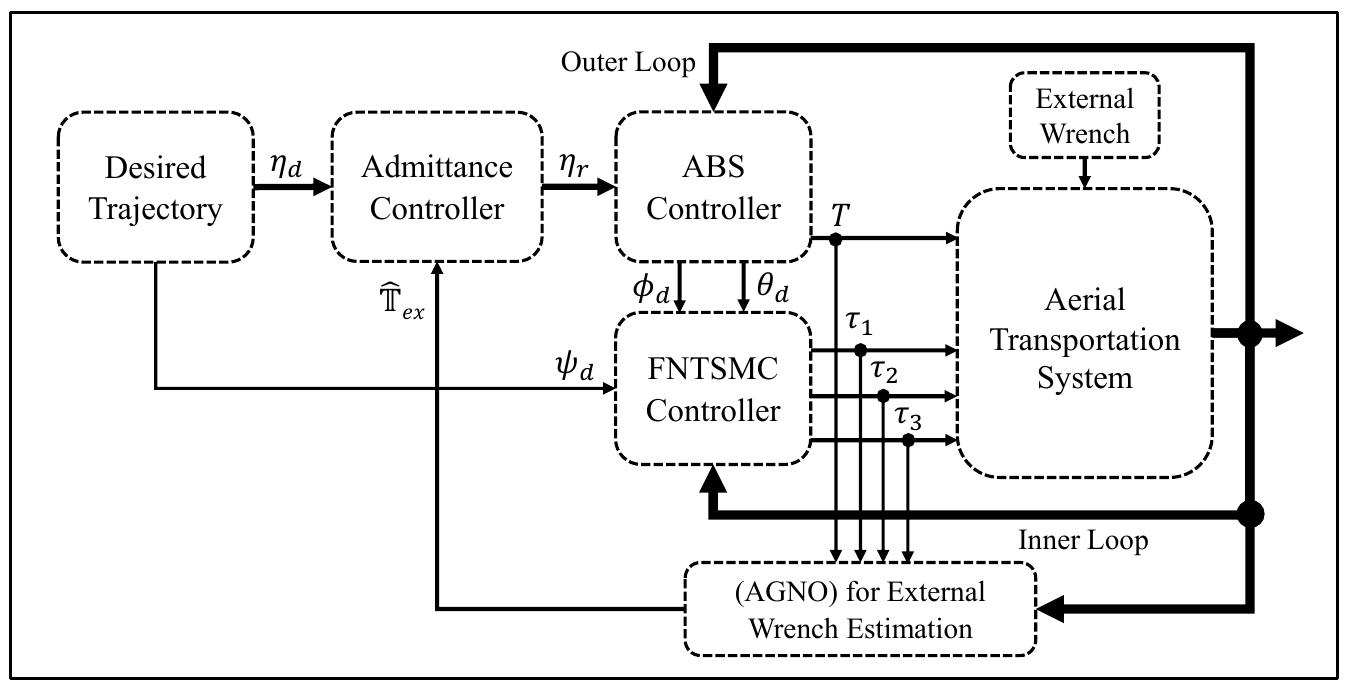}\caption{Control system architecture.}
\label{Fig2:control_SBD}
\end{figure}
While the main focus of this work is the estimation of the external interaction wrenches, a brief description of the control architecture is given below, for more details on the controller design and its stability analysis, see \cite{naser2025aerial, naser2025human}. As illustrated in Figure \ref{Fig2:control_SBD}, the control architecture integrates adaptive backstepping control for position stabilization, Nonsingular Fast Terminal Sliding Mode Control (NFTSMC) for robust attitude corrections, and an admittance controller for natural human-robot interaction. This ensures reliable trajectory tracking, payload stabilization, and responsive human-guided operation. The admittance controller introduces compliance, allowing the UAV to respond adaptively to human-applied forces based on force estimations, enabling precise payload handling \cite{augugliaro2013admittance}.

The control allocation between the total control signals and the $j^{th}$ quadrotor control input is derived as follows:
\begin{equation}
\begin{bmatrix}T\\ \tau\end{bmatrix}=\Gamma u_{q},
\label{eq6:control_allocation}
\end{equation}
where $u_{q}\in\mathbb{R}^{8}$ is the control input vector of both quadrotors:
\begin{equation}
u_{q}=[T_{q_{1}},\tau_{11},\tau_{12},\tau_{13},T_{q_{2}},\tau_{21},\tau_{22},\tau_{23}]^{\top},
\label{eq7:u_o}
\end{equation}
and $\Gamma\in\mathbb{R}^{4\times8}$ is a constant matrix that can be constructed depending on the configuration of the integrated system:
\begin{equation}
\Gamma=\begin{bmatrix}1&0&0&0&1&0&0&0\\ s_{1}(2)&1&0&0&s_{2}(2)&1&0&0\\ -s_{1}(1)&0&1&0&-s_{2}(1)&0&1&0\\ 0&0&0&1&0&0&0&1\end{bmatrix}.
\label{eq8:Gamma_matrix}
\end{equation}
Since the system in \eqref{eq6:control_allocation} is underdetermined, the solution is optimized by minimizing a cost function $\mathfrak{J}(u_{q})$:
\begin{equation}
u_{q}^{*}=\arg\min\{\mathfrak{J}|[T^{des},\tau^{des}]^{\top}=\Gamma u_{q}\},    
\label{eq9:argmin_J}
\end{equation}
where the cost function is:
\begin{equation}
\mathfrak{J}=\sum_{j=1}^{2} d_{j_{1}}T_{q_{j}}^{2}+d_{j_{2}}\tau_{j_{1}}^{2}+d_{j_{3}}\tau_{j_{2}}^{2}+d_{j_{4}}\tau_{j_{3}}^{2}.
\label{eq10:cost_function}
\end{equation}
where $d_{ji}$ are positive constants. Equations \eqref{eq9:argmin_J} and \eqref{eq10:cost_function} can be combined and reformulated as $\mathfrak{J}=\|\mathcal{O}u_{q}\|_{2}^{2}$ where:
\begin{equation}
\mathcal{O}=\sqrt{\text{diag}(d_{11},d_{12},d_{13},d_{14},d_{21},d_{22},d_{23},d_{24})}
\label{eq11:A_matrix_n}
\end{equation}
An optimal control allocation is then computed using the matrices in \eqref{eq8:Gamma_matrix} and \eqref{eq11:A_matrix_n} and utilizing the Moore-Penrose pseudoinverse ($^+$):
\begin{equation}
\begin{aligned}
u_{q}^{*}&=\mathcal{O}^{-1}(\Gamma\mathcal{O}^{-1})^{+}[T^{des},\tau^{des}]^{\top},\\ &=\mathcal{O}^{-2}\Gamma^{\top}(\Gamma\mathcal{O}^{-2}\Gamma^{\top})^{-1}[T^{des},\tau^{des}]^{\top}.
\end{aligned}
\label{eq12:moore_penrose}
\end{equation}
This strategy ensures an optimal distribution of control efforts.

\section{The Proposed Nonlinear Observer Design \label{sec:observer_design}}
For the estimation of the external interaction wrench, we express the dynamic model (\ref{eq5:integrated_dynamics}) in the inertial frame in terms of Euler angles as follows:
\begin{equation}
\begin{aligned}
m\ddot{p}&=R(\xi)T\hat{e}-mg\hat{e}+f_{ex}, \\
J(\xi)\ddot{\xi}&=\tau_I-C_{r}(\xi,\dot{\xi})\dot{\xi}+M_{ex},
\end{aligned}
\label{eq13:euler_dynamics}
\end{equation}
where $R(\xi)$ is the rotation matrix, $\tau_I=R(\xi)\tau$, $C_{r}(\xi,\dot{\xi})$ represents Coriolis and centrifugal terms, and it is given as follows:
\begin{equation}
C_{r}(\xi,\dot{\xi})\dot{\xi}=J\dot{\Theta}(\xi,\dot{\xi})\dot{\xi}+\begin{bmatrix}\Theta(\xi)\dot{\xi}\end{bmatrix}_{\times}\begin{bmatrix}J\Theta(\xi)\dot{\xi}\end{bmatrix},
\label{eq14:C_r_term}
\end{equation}
where $\dot{\Theta}(\xi,\dot{\xi})$ is the time derivative of the matrix in \eqref{eq1:omega_psi_dot}, $[\cdot]_{\times}$ represents the skew symmetric matrix, and $J(\xi)$ is the inertia tensor, which explicitly depends on the attitude $\xi$ such that:
\begin{equation}
J(\xi)=\Theta(\xi)^{\top}I\Theta(\xi) 
  =\begin{bmatrix}J_{11} & J_{12} & J_{13}\\
J_{21} & J_{22} & J_{23}\\
J_{31} & J_{32} & J_{33}
\end{bmatrix}
\label{eq16:inertia_tensor}
\end{equation}
where the internal components of the matrix $J(\xi)$ in \eqref{eq16:inertia_tensor} are given as follows: 
\begin{equation}
\begin{aligned}J_{11} & =I_{xx}\cos^{2}\theta+I_{zz}\sin^{2}\theta\\
J_{12} & =0\\
J_{13} & =(I_{zz}-I_{xx})\cos\phi\sin\theta\cos\theta\\
J_{21} & =J_{12}\\
J_{22} & =I_{yy}\\
J_{23} & =I_{yy}\sin\phi\\
J_{31} & =J_{13}\\
J_{32} & =J_{23}\\
J_{33} & =I_{xx}\cos^{2}\phi\sin^{2}\theta+I_{zz}\cos^{2}\phi\cos^{2}\theta+I_{yy}\sin^{2}\phi
\end{aligned}
\label{eq17:inertia_components}
\end{equation}
The explicit dependence of $J(\xi)$ on $\xi$ is a key aspect of the model, accounting for asymmetric mass distributions and varying payloads that are common in practical aerial system applications.

The dynamic model in (\ref{eq13:euler_dynamics}) can be expressed in a compact form to develop a unified expression for the external interaction wrench:
\begin{equation}
\mathbb{T}_{ex}=\mathbb{M}(\eta)\ddot{\eta}+\mathbb{C}(\eta,\dot{\eta})\dot{\eta}+\mathbb{G}+\mathbb{A}(\eta)u
\label{eq18:compact_dynamics}
\end{equation}
where $\mathbb{T}_{ex}=[f_{ex}^{\top},M_{ex}^{\top}]^{\top}$ is the external interaction wrench, $\eta=[p^{\top},\xi^{\top}]^{\top}$ is the generalized configuration vector, $u=[T,\tau_I^{\top}]^{\top}$ is the control input vector, $\mathbb{M}(\eta)\in\mathbb{R}^{6\times6}$ is the positive definite inertia matrix, $\mathbb{C}(\eta,\dot{\eta})\in\mathbb{R}^{6\times6}$ is the Coriolis matrix, $\mathbb{G}\in\mathbb{R}^{6}$ is the gravity vector and $\mathbb{A}(\eta)\in\mathbb{R}^{6\times4}$ is the matrix of the control input coefficients. These matrices are given as follows:
\begin{equation}
\begin{aligned}\mathbb{M}(\eta)~~= & \begin{bmatrix}m\mathbb{I}_{3\times3} & 0_{3\times3}\\
0_{3\times3} & J(\eta)
\end{bmatrix}, \quad \mathbb{G}~~~= &\begin{bmatrix}mg\hat{e}\\
0
\end{bmatrix},\\
\mathbb{C}(\eta,\dot{\eta})= & \begin{bmatrix}0_{3\times3} & 0_{3\times3}\\
0_{3\times3} & C_{r}(\eta,\dot{\eta})
\end{bmatrix},~\mathbb{A}(\eta)= & \begin{bmatrix}-R(\eta)\hat{e} & 0_{3\times3}\\
0_{3\times1} & -R(\eta)
\end{bmatrix},
\end{aligned}
\label{eq19:matrices_def}
\end{equation}
where $\mathbb{I}_{n\times n}$ and $0_{n\times m}$ are identity and zero matrices, respectively.

\subsection{The Proposed Nonlinear Observer}
Based on \eqref{eq18:compact_dynamics}, a nonlinear observer is proposed to estimate the external interaction wrench as follows:
\begin{equation}
\begin{aligned}
\dot{\hat{\mathbb{T}}}_{ex}= & \quad \mathbb{B}(\eta,\dot{\eta})\big(\mathbb{T}_{ex}-\hat{\mathbb{T}}_{ex}\big),\\
= & -\mathbb{B}(\eta,\dot{\eta})\hat{\mathbb{T}}_{ex}+\mathbb{B}(\eta,\dot{\eta})\Big(\mathbb{M}(\eta)\ddot{\eta}+\mathbb{C}(\eta,\dot{\eta})\dot{\eta}\\
 & \hspace{13.5em}+\mathbb{G}+\mathbb{A}(\eta)u\Big),
\end{aligned}
\label{eq20:observer_design}
\end{equation}
where $\hat{\mathbb{T}}_{ex}=[\hat{f}_{ex}^{\top},\hat{M}_{ex}^{\top}]^{\top}\in\mathbb{R}^{6}$
represents the wrench estimation and $\mathbb{B}(\eta,\dot{\eta})\in\mathbb{R}^{6\times6}$
is a matrix to be designed to guarantee the observer convergence.
 Assuming $\dot{\mathbb{T}}_{ex}=0$ for slowly varying physical interaction wrench, the observer estimation error and its dynamics become:
 \begin{equation}
\begin{aligned}
e&=\mathbb{T}_{ex}-\hat{\mathbb{T}}_{ex}, \\
\dot{e}&=\dot{\mathbb{T}}_{ex}-\dot{\hat{\mathbb{T}}}_{ex}, \\
&=-\mathbb{B}(\eta,\dot{\eta})(\mathbb{T}_{ex}-\hat{\mathbb{T}}_{ex}), \\
&=-\mathbb{B}(\eta,\dot{\eta})e \qquad ~~\implies \quad \dot{e}+\mathbb{B}(\eta,\dot{\eta})e=0.
\label{eq21:error_dynamics}
\end{aligned}
\end{equation}
This error dynamics in \eqref{eq21:error_dynamics} shows that the asymptotic stability of the observer can be guaranteed by an appropriate design of the matrix $\mathbb{B}(\eta,\dot{\eta})$.

\subsection{Acceleration-Free Implementation}
A practical challenge in implementing the proposed observer is the
requirement for acceleration measurements $(\ddot{\eta})$, particularly the angular acceleration $(\ddot{\xi})$, which is typically not directly available from standard sensor suites. To address this issue, we define an auxiliary vector as follows: 
\begin{equation}
\delta=\hat{\mathbb{T}}_{ex}-\Lambda(\dot{\eta}).\label{eq22:auxiliary_vector}
\end{equation}
Differentiating \eqref{eq22:auxiliary_vector} with respect to time and rearranging: 
\begin{equation}
\dot{\hat{\mathbb{T}}}_{ex}=\dot{\delta}+\frac{\partial\Lambda(\dot{\eta})}{\partial\dot{\eta}}\ddot{\eta}.\label{eq23:auxiliary_vector_dot}
\end{equation}
Substituting \eqref{eq22:auxiliary_vector} and \eqref{eq23:auxiliary_vector_dot} into \eqref{eq20:observer_design} results in the following. 
\begin{multline}
\dot{\delta}+\frac{\partial\Lambda(\dot{\eta})}{\partial\dot{\eta}}\ddot{\eta}=-\mathbb{B}(\eta,\dot{\eta})\big(\delta+\Lambda(\dot{\eta})\big)+\mathbb{B}(\eta,\dot{\eta})\\
\Big(\mathbb{M}(\eta)\ddot{\eta}+\mathbb{C}(\eta,\dot{\eta})\dot{\eta}+\mathbb{G}+\mathbb{A}(\eta)u\Big).\label{eq24:observer_auxiliary}
\end{multline}
To eliminate the acceleration term, we choose: 
\begin{equation}
\frac{\partial\Lambda(\dot{\eta})}{\partial\dot{\eta}}=\mathbb{B}(\eta,\dot{\eta})\mathbb{M}(\eta).\label{eq25:lambda_choice}
\end{equation}
Substituting \eqref{eq25:lambda_choice} into \eqref{eq24:observer_auxiliary} and rearranging result in the acceleration-free observer dynamics \eqref{eq26:auxiliary_wrench_observer_dynamics}: 
\begin{equation}
\begin{aligned}
\dot{\delta}~~= & -\mathbb{B}(\eta,\dot{\eta})\delta+\mathbb{B}(\eta,\dot{\eta})\big(\mathbb{C}(\eta,\dot{\eta})\dot{\eta}+\mathbb{G}+\mathbb{A}(\eta)u-\Lambda(\dot{\eta})\big),\\
\hat{\mathbb{T}}_{ex}= & \quad \delta+\Lambda(\dot{\eta}).
\end{aligned}
\label{eq26:auxiliary_wrench_observer_dynamics}
\end{equation}
For practical implementation, we select $\Lambda(\dot{\eta})=k\dot{\eta}$, which implies that:
\begin{equation}
    \mathbb{B}(\eta,\dot{\eta})=k\mathbb{M}^{-1},
    \label{eq27:B_matrix}
\end{equation}
where $k>0$ is the observer gain. This design results in a computationally efficient observer that does not require acceleration measurements, making it suitable for implementation on the aerial platforms with typical sensor configurations.

\section{Stability, Convergence, and Robustness Analysis \label{sec:stability_analysis}}
\subsection{Stability Analysis for Observer Design}
Stability analysis of nonlinear observers is theoretically challenging and typically relies on Lyapunov methods to prove convergence without explicit solutions. While systems with constant inertia allow relatively straightforward analysis, UAV platforms with varying payloads or asymmetric mass distributions introduce time-varying inertial matrices that complicate this process. Nikoobin and Haghighi \cite{nikoobin2009lyapunov} addressed this issue for robotic manipulators by explicitly accounting for inertia variation in their Lyapunov analysis.

\subsection{Assumptions and Preliminaries}

Before proceeding with the stability analysis, we make the following
assumptions. 
\begin{assum}\label{assum:4}The external interaction wrench will not be rapidly varying during the physical interaction between the aerial platform and the human operator, so that $(\dot{\mathbb{T}}_{ex}=0)$ holds true.\end{assum}
\begin{assum}\label{assum:5}The angular velocities in the roll and
pitch directions are kept bounded during the payload transportation mission.\end{assum}
\begin{assum}\label{assum:6}The inertia tensor $J(\eta)$
is positive definite and bounded for all attainable attitudes in $\eta$.\end{assum}
\begin{assum}\label{assum:7}The rate of change of the inertia tensor
is bounded, i.e., $|\dot{J}(\eta)|\leq\gamma$ for some positive
constant $\gamma$.\end{assum}
These assumptions are reasonable for typical human-UAV interaction
scenarios and provide the foundation for our stability analysis.

\subsection{Lyapunov Stability Analysis}
\begin{thm}
\label{theorem:1}Consider the nonlinear estimator of the external
interaction wrench in \eqref{eq26:auxiliary_wrench_observer_dynamics}
and let Assumptions \ref{assum:4} to \ref{assum:7} hold, the design
of $\mathbb{B}(\xi,\dot{\xi})$ in \eqref{eq27:B_matrix}
ensures asymptotic stability of the error dynamics in \eqref{eq21:error_dynamics}
in the sense of Lyapunov such that $\hat{\mathbb{T}}_{ex}\rightarrow \mathbb{T}_{ex}$
over time. 
\end{thm}
\textbf{Proof:}
Let $V_{e}(e,\eta)=e^{\top}\mathbb{M}(\eta)e$ be a Lyapunov function candidate. This function is positive definite since the inertia matrix $\mathbb{M}(\eta)$ is positive definite. Taking the time derivative and substituting $\mathbb{B}(\eta,\dot{\eta})=k\mathbb{M}^{-1}(\xi)$ yields:
\begin{equation}
\begin{aligned}
\dot{V}_{e}&=\dot{e}^{\top}\mathbb{M}(\eta)e+e^{\top}\dot{\mathbb{M}}(\eta)e+e^{\top}\mathbb{M}(\eta)\dot{e},\\
&=-e^{\top}\mathbb{B}^{\top}(\eta,\dot{\eta})\mathbb{M}(\eta)e+e^{\top}\dot{\mathbb{M}}(\eta)e-e^{\top}\mathbb{M}(\eta)\mathbb{B}(\eta,\dot{\eta})e,\\
&=-2ke^{\top}e+e^{\top}\dot{\mathbb{M}}(\eta)e.
\label{eq28:lyapunov_derivative}
\end{aligned}
\end{equation}
The term $e^{\top}\dot{\mathbb{M}}(\eta)e$ in \eqref{eq28:lyapunov_derivative}
represents the effect of the time-varying inertial matrix. To ensure
stability, we need to show that $\dot{V}_{e}(e,\eta)<0$ for all $e\neq0$. The term $e^{\top}\dot{\mathbb{M}}(\eta)e$ can be bounded following that:
\begin{equation}
|e^{\top}\dot{\mathbb{M}}(\eta)e|\leq|e|^{2}|\dot{\mathbb{M}}(\eta)|,\label{equ29:inertia_bound_inquality}
\end{equation}
where $|\dot{\mathbb{M}}(\eta)|$ represents the matrix norm of $\dot{\mathbb{M}}(\eta)$. From the structure of the inertia matrix $\mathbb{M}(\eta)$ in \eqref{eq19:matrices_def},
its time derivative is given as: 
\begin{equation}
\dot{\mathbb{M}}(\eta)=\begin{bmatrix}0_{3\times3} & 0_{3\times3}\\
0_{3\times3} & \dot{J}(\eta)
\end{bmatrix},\label{equ30:inertia_matrix_dot}
\end{equation}
Given Assumption \ref{assum:5}, we have $|\dot{J}(\eta)|\leq\gamma$,
which implies $|\dot{\mathbb{M}}(\eta)|\leq\gamma$, therefore: 
\begin{equation}
|e^{\top}\dot{\mathbb{M}}(\eta)e|\leq\gamma|e|^{2}\label{equ31:inertia_matrix_inquality},
\end{equation}
The result in \eqref{equ31:inertia_matrix_inquality} gives us: 
\begin{equation}
\dot{V}_{e}(e,\eta)\leq-2k|e|^{2}+\gamma|e|^{2}=-(2k-\gamma)|e|^{2},\label{equ32:lyapunov_dot_final}
\end{equation}
For stability, we need $\dot{V}_{e}(e,\eta)<0$, which is satisfied
when: 
\begin{equation}
2k>\gamma.\label{equ33:gain_inquality}
\end{equation}
The condition in \eqref{equ33:gain_inquality} provides a design guideline for selecting the observer
gain $k$. It should be chosen such that $k>\gamma/2$, where $\gamma$
is the bound on the rate of change of the inertia tensor.

\subsection{Convergence Properties}

When the condition $2k>\gamma$ is satisfied, the error dynamics
is asymptotically stable, meaning that $\lim_{t\to\infty}e(t)=0$.
This implies that the estimated interaction wrench $\hat{\mathbb{T}}_{ex}$
converges to the actual interaction wrench $\mathbb{T}_{ex}$ over
time. The convergence rate is determined by the eigenvalues of the
matrix $\mathbb{B}(\eta,\dot{\eta})$. For our design with $\mathbb{B}(\eta,\dot{\eta})=k\mathbb{M}^{-1}(\eta)$,
the convergence rate is characterized by the time constant $\tau_{c}$
such that: 
\begin{equation}
\tau_{c}=\frac{1}{k\lambda_{min}(\mathbb{M}^{-1}(\eta))}=\frac{\lambda_{max}(\mathbb{M}(\eta))}{k},\label{equ34:convergence_time_const}
\end{equation}
where $\lambda_{min}(\cdot)$ and $\lambda_{max}(\cdot)$ denote the
minimum and maximum eigenvalues, respectively. This analysis shows
that higher values of $k$ lead to faster convergence, but there is
a trade-off with noise sensitivity. In practice, the gain $k$ should
be chosen to balance the convergence speed and the robustness to measurement
noise.

\subsection{Robustness Analysis}

In practical scenarios, the system model may contain uncertainties
and external disturbances. Let us denote the actual system dynamics
as follows: 
\begin{equation}
\mathbb{T}_{ex}=\mathbb{M}(\eta)\ddot{\eta}+\mathbb{C}(\eta,\dot{\eta})\dot{\eta}+\mathbb{G}+\mathbb{A}(\eta)u+\Delta(\eta,\dot{\eta}),\label{equ35:system_dynamics_with_uncertainties}
\end{equation}
where $\Delta(\eta,\dot{\eta})$ represents the combined effect of model
uncertainties and external disturbances. According to \eqref{equ35:system_dynamics_with_uncertainties}, the observer error dynamics
in \eqref{eq21:error_dynamics} now becomes: 
\begin{equation}
\begin{aligned}\dot{e} & =-\mathbb{B}(\eta,\dot{\eta})e+\dot{\mathbb{T}}_{ex}-\mathbb{B}(\eta,\dot{\eta})\Delta(\eta,\dot{\eta}),\\
 & =-\mathbb{B}(\eta,\dot{\eta})e-\mathbb{B}(\eta,\dot{\eta})\Delta(\eta,\dot{\eta}).
\end{aligned}
\label{equ36:error_dynamics_with_uncertainties}
\end{equation}
Since $(\dot{\mathbb{T}}_{ex}=0$ according to Assumption \ref{assum:4}. Using the Lyapunov function in \eqref{eq28:lyapunov_derivative}, the time derivative becomes: 
\begin{equation}
\dot{V}_{e}(e,\eta)=-2ke^{\top}e+e^{\top}\dot{\mathbb{M}}(\eta)e-2ke^{\top}\Delta(\eta,\dot{\eta}).
\label{eq37:disturped_lyapunov_dot}
\end{equation}
Assuming bounded $\Delta(\eta,\dot{\eta})$, for some constants $\varepsilon >0$ i.e., $\|\Delta(\eta,\dot{\eta})\|\leq \varepsilon$,
it can be shown that the observer error remains bounded, with the bound depending on the magnitude of this term. This robustness analysis demonstrates that the proposed nonlinear observer maintains bounded estimation errors even in the presence of model uncertainties and external disturbances, a critical property for practical implementation in real-world human-UAV interaction scenarios.

\subsection{Adaptive Gain Design}

To further improve the performance of the nonlinear observer, an adaptive gain design is proposed such that: 
\begin{equation}
\mathbb{B}(\eta,\dot{\eta})=\mathbb{K}(\eta,\dot{\eta})\mathbb{M}^{-1}(\eta),
\label{eq38:adaptivegain_design}
\end{equation}
where $\mathbb{K}(\eta,\dot{\eta})$ is a positive definite gain matrix
that adapts based on the system state such that:
\begin{equation}
\mathbb{K}(\eta,\dot{\eta})=k_{0}\mathbb{I}_{6\times6}+k_{1}\|\dot{\eta}\|\mathbb{I}_{6\times6}+k_{2}\|\mathbb{M}(\eta)\|\mathbb{I}_{6\times6},
\label{eq39:adaptive_gain_matrix}
\end{equation}
where $k_{0},k_{1},k_{2}>0$ are scalar gains and $\|\cdot\|$ denotes
the norm of a vector or a matrix. This adaptive gain design ensures
that:
\begin{enumerate}
\item The observer gain increases with higher velocities, providing faster
convergence during dynamic motions. 
\item The observer gain adapts to changes in the inertial properties, maintaining
stability across different configurations. 
\item A minimum gain $k_{0}$ ensures that basic stability properties are maintained
at all times.
\end{enumerate}

\subsection{Integration with Control Frameworks}

The estimated interaction wrench $\hat{\mathbb{T}}_{ex}$ is fed to the admittance controller to generate the reference trajectories for the aerial platform as follows:
\begin{equation}
M_{a}(\ddot{\eta}_{d}-\ddot{\eta}_{r})+B_{a}(\dot{\eta}_{d}-\dot{\eta}_{r})+K_{a}(\eta_{d}-\eta_{r})=\hat{\mathbb{T}}_{ex},
\label{eq40:admittance_control}
\end{equation}
where $\eta_{d}$ is the desired trajectory (position and attitude),
$\eta_{r}$ is the reference trajectory generated
by the admittance controller based on $\hat{\mathbb{T}}_{ex}$, and $M_{a},B_{a},K_{a}$ are the virtual inertia, damping,
and stiffness matrices that define the interaction behavior. This
admittance control framework allows the UAV to respond compliantly
to the human-applied forces, creating an intuitive interaction experience.

\section{Simulation Results and Discussion \label{sec:simulation_results}}
\subsection{Simulation Setup}
The aerial cooperative payload transportation system, featuring two quadrotors to lift and transport the payload in response to the human-applied guidance. The system and AGNO observer were simulated using MATLAB (2024b). The integrated system was modeled as a single aerial vehicle with parameters detailed in Table \ref{tab:system_parameters}.

\begin{table}[htbp]
\caption{System and Observer parameters.}
\begin{center}
\begin{tabular}{|c||c||c|}
\hline
\textbf{Symbol} & \textbf{Definition} & \textbf{Value/Unit} \\
\hline \hline
$m$ & Total mass & $3.9/\text{kg}$ \\ \hline
$L_{p}$ & Payload length & $2/\text{m}$ \\ \hline
$g$ & Gravitational acceleration & $9.81/\text{m/s}^{2}$ \\ \hline
$I_{xx}$ & Moment of inertia in x-axis & $3.227/\text{kg~m}^{2}$ \\ \hline
$I_{yy}$ & Moment of inertia in y-axis & $0.061/\text{kg~m}^{2}$ \\ \hline
$I_{zz}$ & Moment of inertia in z-axis & $3.277/\text{kg~m}^{2}$ \\ \hline
$k_0$ & Observer gain & $0.78/-$ \\ \hline
$k_1$ & Observer gain & $0.3/-$ \\ \hline
$k_2$ & Observer gain & $0.35/-$ \\ \hline
$T$ & Maximum UAV thrust & $35/\text{N}$ \\ \hline
$t_{s}$ & Control timestep & $0.01/\text{s}$ \\ \hline
$M_{a}$ & Virtual inertial matrix & $0.95\mathbb{I}_{n\times n}/\text{kg}$ \\ \hline
$B_{a}$ & Virtual damping & $1.54\mathbb{I}_{n\times n}/\text{N~s/m}$ \\ \hline
$K_{a}$ & Virtual spring & $0\mathbb{I}_{n\times n}/\text{N~s/m}$ \\
\hline
\end{tabular}
\label{tab:system_parameters}
\end{center}
\end{table}

\subsection{Simulation Results}
During simulation, the system successfully lifted and transported the payload in response to human guidance. It followed complex 3D trajectories while maintaining stable tracking and attitude control, with initial errors rapidly converged to zero due to accurate estimation of the applied forces and torques, as shown in Figure ~\ref{Fig3:main_matlab}. The arrows indicate the magnitude and direction of the external interaction forces and torques that guided the payload from the initial position to the final destination.

\begin{figure}[b!]
\centering \includegraphics[clip,width=0.9\columnwidth]{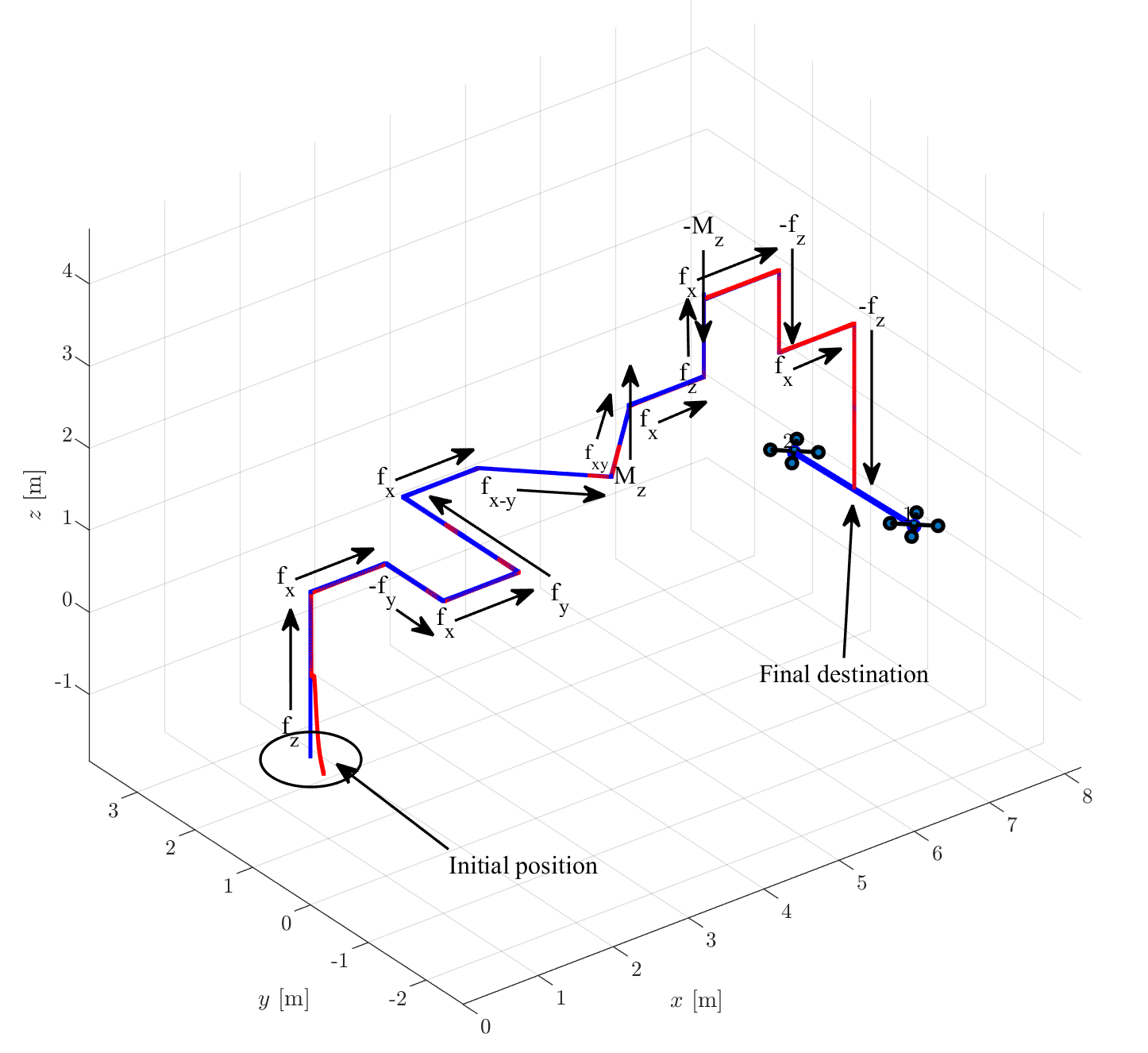}
\caption{The simulation environment; MATLAB simulation of the aerial payload
transport system's motion following human guidance in 3D space as
depicted by the arrows (magnitudes and directions).}
\label{Fig3:main_matlab} 
\end{figure}

The proposed AGNO demonstrated precise estimations of the external interaction wrenches, with force and torque errors approaching zero, as shown in Figures~\ref{Fig4:force_estimation} and~ \ref{Fig5:force_estimation_error}. The estimated and actual wrenches are closely aligned across all axes, confirming the observer's effectiveness in capturing the interaction dynamics. Combined with admittance control, this enables natural and stable human-UAV physical collaboration.

\begin{figure}[b!]
\centering \includegraphics[clip,width=0.9\columnwidth]{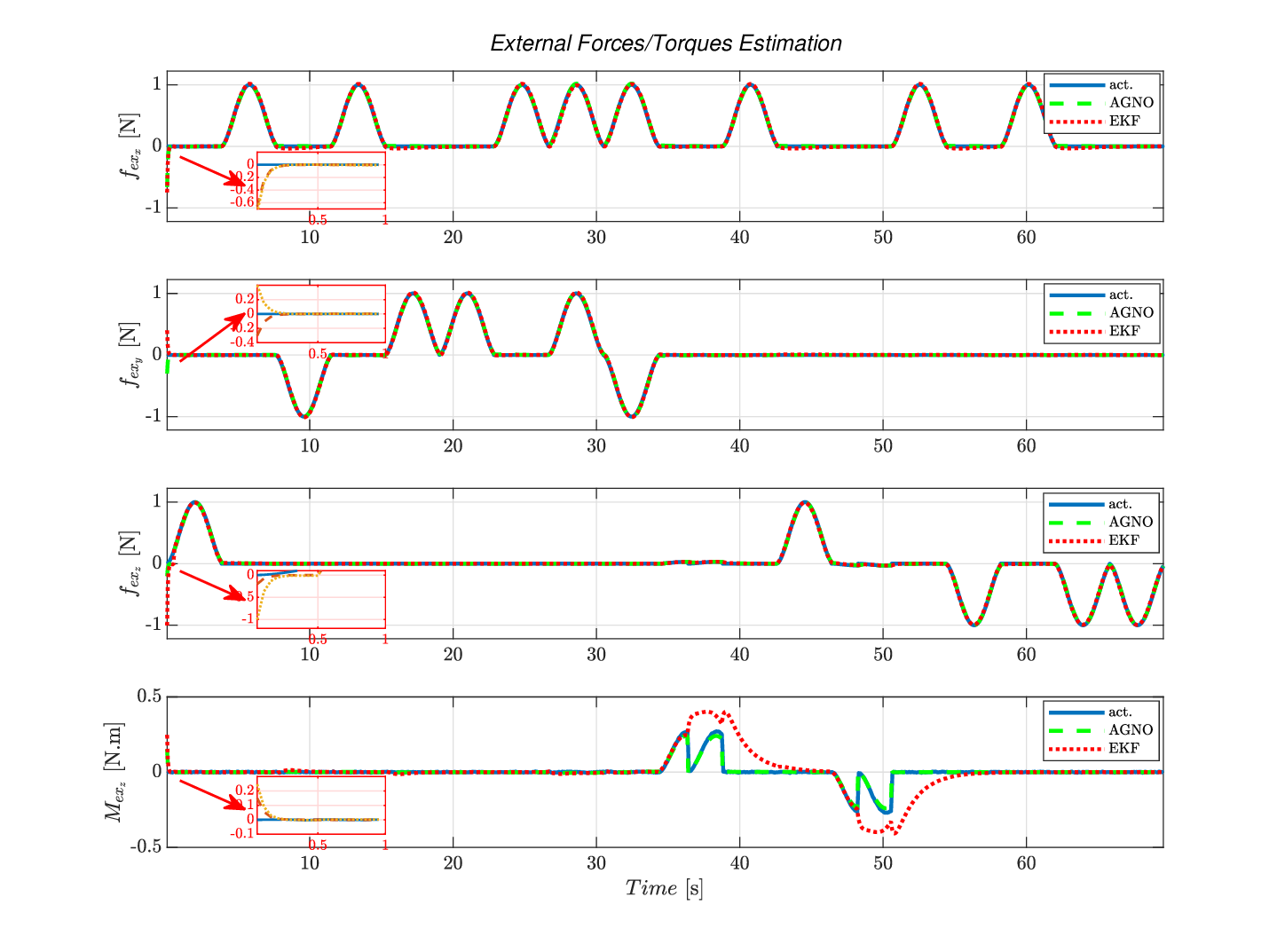}
\caption{The actual and estimated forces in $(x,y,z)$ directions and torque
around $z$ axis.}
\label{Fig4:force_estimation} 
\end{figure}

\begin{figure}[bt!]
\centering \includegraphics[clip,width=0.9\columnwidth]{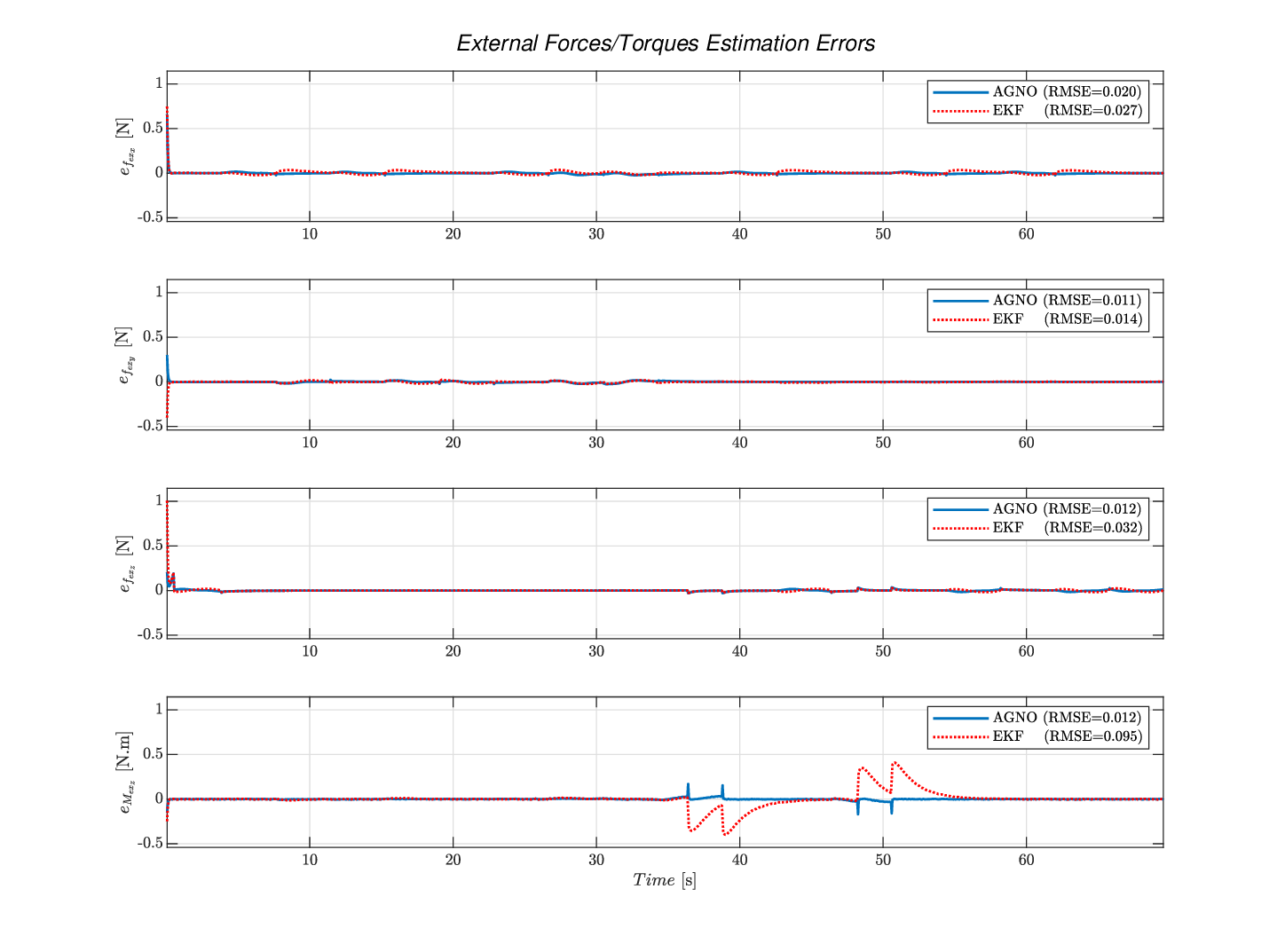}
\caption{Estimation errors of the external interaction forces in $(x,y,z)$
directions and torque around $z$ axis.}
\label{Fig5:force_estimation_error} 
\end{figure}

Figure~\ref{Fig5:force_estimation_error} further quantifies this performance over the 70-seconds simulation. Low estimation errors were maintained during rapid initialization (0-2 seconds), direction reversals (8-12 seconds), bidirectional motion (32-36 seconds), and coupled force-torque interactions (43-56 seconds), highlighting the robustness of the proposed observer under highly nonlinear dynamics and interaction conditions. This eliminates the need for bulky and expensive force sensors, reducing system weight and cost, and enabling natural and intuitive human-UAV physical interaction at any point of contact.

To benchmark the performance of the proposed AGNO, an Extended Kalman Filter (EKF) was implemented under the same conditions. Although the EKF provided reasonable force estimation during smooth motion, its performance degraded significantly during rapid direction changes and coupled force-torque interactions, particularly in torque estimation, where the limitations of model linearization became evident. Quantitatively, the proposed AGNO outperformed the EKF in terms of RMSE for both force and torque estimation. In particular, it achieved RMSE values of (0.020, 0.011, 0.012, and 0.012) compared to (0.027, 0.014, 0.032, and 0.095) for the EKF in estimating the external forces along the $(x,y,z)$ directions and the torque around the $z$-axis, respectively, as shown in Figure~\ref{Fig5:force_estimation_error}. These results demonstrate the advantage of the proposed AGNO in handling the strong nonlinearities of the system dynamics without relying on local linearization, resulting in more reliable and accurate wrench estimation than the EKF.

\section{Conclusion \label{sec:conclusion}}
This paper presented an Adaptive Gain Nonlinear Observer (AGNO) for estimating the external interaction wrench in human-UAV physical interaction. By explicitly accounting for the time-varying inertial matrix, the proposed AGNO maintains high estimation accuracy during dynamic maneuvers and payload variations. The Lyapunov-based analysis provided theoretical guarantees of convergence and robustness. Benchmarking against an Extended Kalman Filter (EKF) demonstrated that the proposed AGNO achieved lower RMSE values in force estimation and a superior torque estimation, particularly under highly nonlinear interaction conditions where EKF performance degrades due to linearization. The proposed method enables intuitive and safe interaction without force-torque sensors, reducing system weight, cost, and complexity while preserving system performance. This approach supports advanced applications such as collaborative aerial transportation and manipulation.

\bibliographystyle{IEEEtran}
\bibliography{references}

@article{laghari2023unmanned,
  title={Unmanned aerial vehicles: A review},
  author={Laghari, Asif Ali and Jumani, Awais Khan and Laghari, Rashid Ali and Nawaz, Haque},
  journal={Cognitive Robotics},
  volume={3},
  pages={8--22},
  year={2023},
  publisher={Elsevier}
}

@article{naser2025aerial,
  title={Aerial assistive payload transportation using quadrotor UAVs with nonsingular fast terminal SMC for human physical interaction},
  author={Naser, Hussein N and Hashim, Hashim A and Ahmadi, Mojtaba},
  journal={Results in Engineering},
  volume={25},
  pages={103701},
  year={2025},
  publisher={Elsevier}
}

@article{naser2026QUKF,
title = {Quaternion-based unscented Kalman filter for robust wrench estimation of human-UAV physical interaction},
journal = {Signal Processing},
volume = {245},
pages = {110582},
year = {2026},
publisher={Elsevier},
issn = {0165-1684},
doi = {https://doi.org/10.1016/j.sigpro.2026.110582},
url = {https://www.sciencedirect.com/science/article/pii/S0165168426000964},
author = {Hussein N. Naser and Hashim A. Hashim and Mojtaba Ahmadi}
}

@inproceedings{hogan1984impedance,
  title={Impedance control: An approach to manipulation},
  author={Hogan, Neville},
  booktitle={1984 American control conference},
  pages={304--313},
  year={1984},
  organization={IEEE}
}

@inproceedings{mariotti2019admittance,
  title={Admittance control for human-robot interaction using an industrial robot equipped with a F/T sensor},
  author={Mariotti, Eleonora and Magrini, Emanuele and De Luca, Alessandro},
  booktitle={2019 International Conference on Robotics and Automation (ICRA)},
  pages={6130--6136},
  year={2019},
  organization={IEEE}
}

@inproceedings{augugliaro2013admittance,
  title={Admittance control for physical human-quadrocopter interaction},
  author={Augugliaro, Federico and D'Andrea, Raffaello},
  booktitle={2013 European Control Conference (ECC)},
  pages={1805--1810},
  year={2013},
  organization={IEEE}
}

@article{ruggiero2018aerial,
  title={Aerial manipulation: A literature review},
  author={Ruggiero, Fabio and Lippiello, Vincenzo and Ollero, Anibal},
  journal={IEEE Robotics and Automation Letters},
  volume={3},
  number={3},
  pages={1957--1964},
  year={2018},
  publisher={IEEE}
}

@article{lee2018integrated,
  title={An integrated framework for cooperative aerial manipulators in unknown environments},
  author={Lee, Hyeonbeom and Kim, Hyoin and Kim, Woojin and Kim, H Jin},
  journal={IEEE Robotics and Automation Letters},
  volume={3},
  number={3},
  pages={2307--2314},
  year={2018},
  publisher={IEEE}
}

@article{rajappa2017design,
  title={Design and implementation of a novel architecture for physical human-UAV interaction},
  author={Rajappa, Sujit and B{\"u}lthoff, Heinrich and Stegagno, Paolo},
  journal={The International Journal of Robotics Research},
  volume={36},
  number={5-7},
  pages={800--819},
  year={2017},
  publisher={SAGE Publications Sage UK: London, England}
}

@inproceedings{rajappa2017control,
  title={A control architecture for physical human-UAV interaction with a fully actuated hexarotor},
  author={Rajappa, Sujit and B{\"u}lthoff, Heinrich H and Odelga, Marcin and Stegagno, Paolo},
  booktitle={2017 IEEE/RSJ International Conference on Intelligent Robots and Systems (IROS)},
  pages={4618--4625},
  year={2017},
  organization={IEEE}
}

@inproceedings{skrede2024linear,
  title={A Linear Discrete Kalman Filter to Estimate the Contact Wrench of an Unknown Robot End Effector},
  author={Skrede, Aleksander},
  booktitle={2024 IEEE International Conference on Real-time Computing and Robotics (RCAR)},
  pages={341--347},
  year={2024},
  organization={IEEE}
}

@inproceedings{yin2023error,
  title={Error-State Kalman Filter Based External Wrench Estimation for MAVs Under a Cascaded Architecture},
  author={Yin, Yuhan and Yang, Qingkai and Fang, Hao},
  booktitle={2023 IEEE/RSJ International Conference on Intelligent Robots and Systems (IROS)},
  pages={5019--5026},
  year={2023},
  organization={IEEE}
}

@article{banks2021physical,
  title={Physical human-UAV interaction with commercial drones using admittance control},
  author={Banks, Christopher and Bono, Antonio and Coogan, Samuel},
  journal={IFAC-PapersOnLine},
  volume={54},
  number={20},
  pages={258--264},
  year={2021},
  publisher={Elsevier}
}

@article{wilmsen2019nonlinear,
  title={Nonlinear wrench observer design for an aerial manipulator},
  author={Wilmsen, Marek and Yao, Chao and Schuster, Micha and Li, Shixiong and Janschek, Klaus},
  journal={IFAC-PapersOnLine},
  volume={52},
  number={22},
  pages={1--6},
  year={2019},
  publisher={Elsevier}
}

@article{nikoobin2009lyapunov,
  title={Lyapunov-based nonlinear disturbance observer for serial n-link robot manipulators},
  author={Nikoobin, Amin and Haghighi, Reza},
  journal={Journal of Intelligent and Robotic Systems},
  volume={55},
  pages={135--153},
  year={2009},
  publisher={Springer}
}

@article{veil2021nonlinear,
  title={Nonlinear disturbance observers for robotic continuum manipulators},
  author={Veil, Carina and Mueller, Daniel and Sawodny, Oliver},
  journal={Mechatronics},
  volume={78},
  pages={102518},
  year={2021},
  publisher={Elsevier}
}

@article{kruvzic2021end,
  title={End-effector force and joint torque estimation of a 7-DoF robotic manipulator using deep learning},
  author={Kru{\v{z}}i{\'c}, Stanko and Musi{\'c}, Josip and Kamnik, Roman and Papi{\'c}, Vladan},
  journal={Electronics},
  volume={10},
  number={23},
  pages={2963},
  year={2021},
  publisher={MDPI}
}

@inproceedings{dai2019learning,
  title={Learning-based external wrench estimation for quadrotors},
  author={Dai, Yi-Wei and Ye, Wei-Yuan and Pi, Chen-Huan and Cheng, Stone},
  booktitle={2019 International Conference on Advanced Mechatronic Systems (ICAMechS)},
  pages={245--249},
  year={2019},
  organization={IEEE}
}

@inproceedings{alharbat2025external,
  title={External-Wrench Estimation for Aerial Robots Exploiting a Learned Model},
  author={Alharbat, Ayham and Ruscelli, Gabriele and Diversi, Roberto and Mersha, Abeje},
  booktitle={2025 International Conference on Unmanned Aircraft Systems (ICUAS)},
  pages={323--331},
  year={2025},
  organization={IEEE}
}

@inproceedings{yuksel2014nonlinear,
  title={A nonlinear force observer for quadrotors and application to physical interactive tasks},
  author={Y{\"u}ksel, Burak and Secchi, Cristian and B{\"u}lthoff, Heinrich H and Franchi, Antonio},
  booktitle={2014 IEEE/ASME international conference on advanced intelligent mechatronics},
  pages={433--440},
  year={2014},
  organization={IEEE}
}

@INPROCEEDINGS{naser2025human,
  author={Naser, Hussein N. and Hashim, Hashim A. and Ahmadi, Mojtaba},
  booktitle={2025 American Control Conference (ACC)}, 
  title={Human Physical Interaction based on UAV Cooperative Payload Transportation System using Adaptive Backstepping and FNTSMC}, 
  year={2025},
  volume={},
  number={},
  pages={450-456},
  doi={10.23919/ACC63710.2025.11107760}}

@article{hashim2023observer,
	title={{O}bserver-based {C}ontroller for {VTOL}-{UAV}s {T}racking using {D}irect {V}ision-Aided {I}nertial {N}avigation {M}easurements},
author={Hashim, Hashim A and Eltoukhy, Abdelrahman EE and Odry, Akos},
journal={ISA transactions},
volume={137},
pages={133--143},
year={2023},
publisher={Elsevier}
}

@inproceedings{tomic2014unified,
  title={A unified framework for external wrench estimation, interaction control and collision reflexes for flying robots},
  author={Tomi{\'c}, Teodor and Haddadin, Sami},
  booktitle={2014 IEEE/RSJ international conference on intelligent robots and systems},
  pages={4197--4204},
  year={2014},
  organization={IEEE}
}

\end{document}